\documentclass[sigconf]{acmart}
\AtBeginDocument{%
  }

\usepackage{dsfont}
\usepackage{multirow}
\usepackage{algorithm}
\usepackage[noend]{algpseudocode}
\usepackage{subcaption}

\newtheorem{assumption}{Assumption}
\newtheorem{remark}{Remark}

\newtheorem{problem}{Problem}

\newcommand{\reals}{\mathbb{R}}
\newcommand{\naturals}{\mathbb{N}}

\newcommand{\Hk}{\mathcal{H}_k}
\newcommand{\HK}{\mathcal{H}_K}
\newcommand{\Pcal}{\mathcal{P}}
\newcommand{\Dcal}{\mathcal{D}}
\newcommand{\Tau}{\mathcal{T}}

\newcommand{\Ucal}{\mathcal{U}}
\newcommand{\MMD}{\mathrm{MMD}}
\newcommand{\x}{\boldsymbol{x}}
\newcommand{\w}{\boldsymbol{w}}
\renewcommand{\path}{\omega}
\newcommand{\pathX}{\path_{x}}
\newcommand{\pathXbold}{\boldsymbol{\path}_{x}}
\newcommand{\PathX}{\Omega_{x}}

\newcommand{\policy}{\pi_x}

\newcommand{\strategy}{\pi_s}
\newcommand{\Strategy}{\Pi_s}
\newcommand{\adversary}{\xi}
\newcommand{\Adversary}{\Xi}

\newcommand{\pathrmdp}{\omega}

\newcommand{\Pathrmdp}{\Omega}

\newcommand{\last}{\mathrm{last}}

\newcommand{\safe}{\mathrm{safe}}

\newcommand{\Xreach}{X_{\text{reach}}}
\newcommand{\Xavoid}{X_{\text{avoid}}}
\newcommand{\Xsafe}{X_\safe}

\newcommand{\Sreach}{S_{\text{reach}}}

\newcommand{\ig}[1]{\textcolor{blue}{[IG: #1]}}

\begin{document}

\title{
Data-Driven Control via Conditional Mean Embeddings: Formal Guarantees via Uncertain MDP Abstraction
}



\author{Ibon Gracia}
\authornote{Corresponding author}
\email{ibon.gracia@colorado.edu}
\affiliation{%
  \institution{University of Colorado Boulder}
  \city{Boulder}
  \state{Colorado}
  \country{USA}
}

\author{Morteza Lahijanian}
\affiliation{%
  \institution{University of Colorado Boulder}
  \city{Boulder}
  \state{Colorado}
  \country{USA}
}\email{morteza.lahijanian@colorado.edu}


\renewcommand{\shortauthors}{Gracia et al.}

\begin{abstract}

    Controlling stochastic systems with unknown dynamics and under complex specifications is specially challenging in safety-critical settings, where performance guarantees are essential. We propose a data-driven policy synthesis framework that yields formal performance guarantees for such systems using conditional mean embeddings (CMEs) and uncertain Markov decision processes (UMDPs). From trajectory data, we learn the system's transition kernel as a CME, then construct a finite-state UMDP abstraction whose transition uncertainties capture learning and discretization errors. Next, we generate a policy with formal performance bounds through robust dynamic programming. We demonstrate and empirically validate our method through a temperature regulation benchmark.
\end{abstract}

\keywords{Stochastic Systems, Safety-Critical Systems, Data-Driven Control, Kernel Methods, Finite Abstraction, Uncertain MDPs}


\maketitle

\section{Introduction}
\label{sec:intro}


As autonomous dynamical systems continue to be deployed in real-world environments, providing guarantees on their performance becomes paramount. These systems are typically subject to uncertainty, e.g., random external disturbances, and accurately modeling their dynamics is often infeasible due to complex physics and environmental factors.
Consequently, \textit{data-driven control}, which seeks to synthesize controllers (policies) using trajectory data to compensate for the lack of accurate models, has emerged as an essential research direction. However, ensuring that such policies satisfy \emph{safety} requirements remains a major challenge.
This work aims to address this challenge by developing an end-to-end formal framework for data-driven control synthesis with complex requirements that provides rigorous guarantees on system requirements.


Kernel methods are powerful tools for data-driven decision-making. They are especially attractive for safety-critical settings because they can quantify the gap between learned and true models, allowing designers to provide formal statements about the safety of unknown dynamical systems. 
A popular class of such methods is \emph{Gaussian process} (GP) regression, which is employed by
many recent works, e.g., \cite{skovbekk2021formal,schon2024data,reed2025learning}, to certify systems with unknown vector fields and additive noise. 
Another important class is \emph{conditional mean embeddings} (CMEs), which learn the transition kernel of an unknown stochastic system by learning its embedding into a \emph{reproducing kernel Hilbert space} (RKHS). 
Understanding the convergence of the empirical CME to the true one is an active area of research \cite{li2022optimal,grunewalder2012modelling,talwai2022sobolev}.
Specifically, \cite{li2022optimal} shows that the empirical CME converges to the true one at a rate $\mathcal{O}(\log(N)/N)$), with $N$ being the number of samples. 
Unfortunately, existing convergence results are either asymptotic \cite{park2020measure}, rely on quantities that are hard to bound \cite{li2022optimal,talwai2022sobolev}, or assume finite state spaces \cite{grunewalder2012modelling}.

Existing control-synthesis methods based on CMEs generally lack formal guarantees. Specifically, \cite{romao2023distributionally,thorpe2019model} approximate value functions via kernel regression. That approach also face fundamental limitations: when specifications involve reachability or temporal logic, the initial value function is an indicator function, which cannot be uniformly approximated by functions in an RKHS induced by a continuous kernel. 
%
Other works \cite{nemmour2022maximum} consider distributionally robust optimization using ambiguity sets derived from samples. To ensure soundness, \cite{nemmour2022maximum} employs a \textit{maximum mean discrepancy} (MMD) ambiguity set and a \emph{conditional-value-at-risk} relaxation. It however introduces irreducible conservatism.

A CME-based work that provides rigorous guarantees for unknown systems is \cite{schon2024data}. It uses data to construct a CME and verifies safety via barrier certificates by ensuring that the true CME lies within an ambiguity set and by bounding approximation error through a GP surrogate. 
However, it is limited to safety verification and does not extend to control synthesis under complex objectives.


A common method to specify complex system requirements is through \emph{linear temporal logic} (LTL)~\cite{baier2008principles} and its variants, 
and significant success has been achieved in synthesizing policies for such specifications using \emph{finite-abstraction} techniques for stochastic systems with complex dynamics~\cite{gracia2025beyond}, uncertain dynamics~\cite{skovbekk2021formal}, and unknown disturbance~\cite{nazeri2025data,gracia2025efficient,Gracia:L4DC:2025}. 
These abstractions are typically variants of \emph{Markov Decision Processes} (MDPs), such as interval MDPs \cite{lahijanian2015formal} and \emph{uncertain MDPs} (UMDPs)~\cite{iyengar2005robust}. 
In kernel methods, IMDP abstractions have been used for GP regressed models~\cite{skovbekk2021formal,schon2024data,reed2025learning}, but
no formal approach exists for abstracting CMEs.

In this work, we present a novel data-driven control (policy) synthesis framework based on CMEs and UMDP abstraction with formal guarantees. We first learn the empirical CME of the system's transition kernel and then abstract it into a finite UMDP, quantifying the learning gap. Then, we compute a policy on the UMDP through \emph{robust dynamic programming} (RDP), yielding performance guarantees for the underlying (unknown) system. Our contributions are $4$-fold: 
(i) the \emph{first} data-driven control synthesis framework via CMEs that provides formal guarantees, 
(ii) a novel method of UMDP abstraction of a CME model, 
(iii) a concentration inequality for the empirical CME with explicit constants, 
and 
(iv) relaxation of the common assumption that samples must be obtained at prescribed points, instead allowing sampling from an arbitrary distribution.  

\section{Basic Notation}

Given a connected and compact subset $s \subset \reals^n$, we denote by $c_s$ its center point. For a measurable space $(X, \mathcal{F}(X))$, with $\mathcal{F}(X)$ being the Borel $\sigma$-algebra on $X$, we denote by $\Pcal(X)$ the set of Borel probability distributions (measures) on $X$. We let $\delta_x \in \Pcal(X)$ be the Dirac measure on $x\in X$. We use bold symbols to denote random variables, like $\x$, as opposed to points in the sample set $x \in X$.
\section{Problem Formulation}
\label{sec:problem}

We consider a controlled discrete-time stochastic process $\x_t \in \reals^n$ at each time step $t\in \mathbb{N}_0$ with controls $u_t$ taking values in the finite set $U$. Let $\mathcal{B}(\reals^n)$ be the Borel $\sigma$-algebra on $\reals^n$ with respect to the usual topology. The system evolves according to the \emph{transition kernel} $\Tau: \mathcal{B}(\reals^n) \times \reals^n \times U \rightarrow [0,1]$, which gives the probability that the successor state of $x_t \in \reals^n$ under control $u_t \in U$ belongs to the Borel set $B \in \mathcal{B}(\reals^n)$, i.e.,
\begin{align}
\label{eq:sys}
    \x_{t+1} \sim \Tau(\cdot \mid x_t, u_t),
\end{align}
%

Given $x_0,\ldots,x_T \in \reals^n$, $u_0,\dots,u_{T-1}\in U$, and $T \in \mathbb{N}_0$, we let $\pathX = x_0 \xrightarrow{u_0} \ldots \xrightarrow{u_{T-1}} x_T$ denote a finite \emph{trajectory} of System~\eqref{eq:sys}, with length
$|\pathX| = T+1$.
We let $\PathX$ be the set of finite trajectories 
and $\pathX(t)$ be the state of $\pathX$ at time $t$. 
A \emph{policy} $\policy: \PathX \to U$ of System~\eqref{eq:sys} maps each finite trajectory $\pathX$ to a control $\policy(\pathX) \in U$. 

Given a policy $\policy$ and an initial condition $x_0\in \reals^n$, the transition kernel $\Tau$ induces a unique probability measure over the system's trajectories, 
which uniquely extends to the set of trajectories of infinite length \cite{bertsekas1996stochastic}. For the sake of conciseness, we also denote this measure by $P_{x_0}^{\policy}$. Then, $P_{x_0}^{\policy}[\pathXbold(t) \in B]$ 
denotes the probability that $\pathXbold(t)$ belongs to the set $B\in \mathcal{B}(\reals^n)$ when following strategy $\policy$ and starting at $x_0$.

Our objective is to obtain a policy for System~\eqref{eq:sys} that satisfies a complex temporal specification over a given set of regions in $\reals^n$ with high probability. These requirements and often expressed as temporal logic formulas (e.g., LTL), which are equivalent to \emph{reach-avoid} properties over an extended state space. For the sake of simplicity of presentation, we focus on reach-avoid properties.
Given sets $\Xreach, \Xavoid \subseteq \reals^n$ with $\Xreach \subseteq (\reals^n\setminus \Xavoid) =: \Xsafe$, we call $\varphi_x \equiv (\neg \Xavoid \mathcal{U} \Xreach)$ a \textit{reach-avoid} specification, which reads, ``\textit{do not visit $\Xavoid$ until $\Xreach$ is reached.}'' 
The probability that System~\eqref{eq:sys} satisfies the specification $\varphi_x$ under policy $\policy$ when the initial state is $x_0 \in \reals^n$ is defined as
\begin{multline}
    \label{eq: satisfaction_prob}
    \text{Pr}_{x_0}^{\policy}[\varphi_x] = \text{Pr}_{x_0}^{\policy}\big(\big\{\pathXbold \in \PathX \mid \exists t \in \mathbb{N}_0 \text{ s.t. }   
     \pathXbold(t)\in\Xreach \\
     \land \; \forall \,t' \le t, \; \pathXbold(t')\notin \Xavoid \big\} \big),
\end{multline}
%

In this work, we consider the setting where transition kernel $\Tau$ is \emph{unknown}, but we have access to a method of generating datasets of i.i.d. tuples of the form $\boldsymbol{\Dcal}_u := \{(\hat\x^{(i)}, u, \hat\x_+^{(i)})\}_{i = 1}^N$, representing state-control-next-state, for each $u \in U$.

\begin{problem}[Formal Control Synthesis]
\label{prob:problem}
    Let the transition kernel $\Tau$ of System~\eqref{eq:sys} be unknown. Assume we are able to generate datasets $\boldsymbol{\Dcal}_u := \{(\hat\x^{(i)}, u, \hat\x_+^{(i)})\}_{i = 1}^N$, 
    where $\hat\x^{(i)}$ is sampled i.i.d. from some distribution $P$ and $\hat\x_+^{(i)}$ is sampled from $\Tau(\cdot \mid \hat x^{(i)},u)$, 
    for every $u \in U$.\footnote{For simplicity, we use the same distribution $P$ for each $u$, though this is not required.}
    Given a reach-avoid specification $\varphi_x \equiv (\Xreach, \Xavoid)$, 
    synthesize a policy $\policy$ and high probability functions $\underline p, \overline p : \reals^n\rightarrow [0, 1]$ such that
    \begin{align*}
        \forall x_0 \in \reals^n, \quad \text{Pr}_{x_0}^{\policy}[\varphi_x] \in [\underline p(x_0), \overline p(x_0)],
    \end{align*}
    with confidence at least $1 - \delta$ over the samples in the datasets.
\end{problem}

Our approach first leverages the datasets to learn the \emph{conditional mean embedding} (CME) of the transition kernel in~\eqref{eq:sys}. We then construct a finite abstraction in the form of a UMDP, 
whose uncertain transition kernel accounts for both state-space discretization and the learning gap.
Finally, we synthesize a strategy on the UMDP and refine it into a policy for System~\eqref{eq:sys}, ensuring that the reach–avoid probability is guaranteed to lie within the computed interval.

\section{Preliminaries}
\label{sec:preliminaries}

\subsection{Conditional Mean Embeddings}

In this section we give a brief overview of the relevant kernel learning topics. We show how to embed the transition kernel $\Tau$ into a Hilbert space and how to empirically estimate this embedding from data, ensuring ``closeness'' to the true embedding. This allows us to learn a sound model of the unknown $\Tau$ for policy synthesis.


Let $k : \reals^n\times \reals^n \rightarrow \reals$ 
be a measurable, bounded, and symmetric function that satisfies the positive semidefinite condition, i.e., for all $x_1, \dots, x_l \in \reals^n$ and $\alpha_1, \dots, \alpha_l \in\reals$, 
$\sum_{i,j=1}^l \alpha_i\alpha_j k(x_i, x_j) \ge 0$. Then, we say that $k$ is a positive semidefinite (PSD) kernel, and it gives rise to a \emph{reproducing Kernel Hilbert space} (RKHS) \cite{berlinet2011reproducing}:
%
    $\Hk := \overline{\mathrm{span}}\{k(\cdot, x) :\reals^n\rightarrow \reals \mid x\in \reals^n\},$
%
i.e., the closure of the span of the vectors of the form $k(\cdot, x)$. The RKHS $\Hk$ has inner product $\langle k(\cdot, x), k(\cdot, x') \rangle_{\Hk} = k(x, x')$ which extends to all of $\Hk$ via linearity and completion. The RKHS $\Hk$ exhibits a reproducing property,
according to which function evaluation is equivalent to inner product in the RKHS, and therefore a linear and continuous functional, i.e., for all $h \in \Hk$ and $x \in \reals^n$, it holds that $h(x) = \langle h, k(\cdot, x') \rangle_{\Hk}$. The RKHS inner product induces the norm $\|h\|_{\Hk} := \sqrt{\langle h, h \rangle_{\Hk}}$.

While RKHSs are often used as a framework for learning a function as an element of the RKHS, we need a \emph{vector-valued RKHS} (vRKHS) \cite{carmeli2010vector} to learn a transition kernel. Being a common choice in the literature, we consider the operator-valued kernel $K$ such that $K(x_1,x_2) := k(x_1, x_2)\mathrm{Id}$, where $\mathrm{Id}$ is the identity operator on $\Hk$. Nevertheless, we highlight that other choices are also possible. 
Given $x\in \reals^n$ and $h \in \Hk$, $K(\cdot, x)h$ defines a function from $\reals^n$ to $\Hk$. 
The operator-valued kernel $k$ defines the vRKHS
%
    $\HK := \overline{\mathrm{span}}\{K(\cdot, x)h :\reals^n\rightarrow \Hk \mid x\in \reals^n, h \in \Hk\}$
%
equipped with the inner product
$\langle K(\cdot, x_1)h_1, K(\cdot, x_2)h_2 \rangle_{\HK} := \langle h_1, K(x_1, x_2)h_2 \rangle_{\Hk}$, which extends to all of $\HK$ via linearity and completion, and with norm $\|\cdot\|_{\HK} := \sqrt{\langle \cdot, \cdot \rangle_{\HK}}$. Similarly to $\Hk$, $\HK$ exhibits the following reproducing property: for all $H\in \HK$, $h \in \Hk$, and $x \in \reals^n$, it holds that $\langle H(x), h\rangle_{\Hk} = \langle H, K(\cdot, x)h\rangle_{\HK}$.

Leveraging the theory on RKHSs and vRKHSs we can embed the transition kernel $\Tau$ into the RKHS $\HK$ in order to learn the embedding from the system's trajectory data. We denote by $\Psi : \Pcal(\reals^n) \rightarrow \Hk$ the \emph{kernel mean embedding} (ME) operator, which maps every $P \in \Pcal(\reals^n)$ to an RKHS function $\Psi(p) := \mathbb{E}_{x \sim P}[k(\cdot, x)] \in \Hk$. Let 
$\MMD(P, P') := \|\Psi(P) - \Psi(P')\|_{\Hk}$ denote the\emph{maximum mean discrepancy} (MMD) between two probability distributions $P, P' \in \Pcal(\reals^n)$, which metrizes the weak topology of probability measures if $\Psi$ is injective\footnote{This is the case for a wide variety of kernels, including the Gaussian one.}. We also define the conditional mean embedding (CME) of $\Tau$ under $u \in U$ as the function $\mu_u: \reals^n \rightarrow \Hk$ with 
\begin{align*}
    \mu_u(x)(\cdot) := \Psi(\Tau(\cdot\mid x, u)) = \mathbb{E}_{\x_+\sim \Tau(\cdot\mid x, u)}[k(\cdot, \x_+)] \qquad \forall x \in \reals^n.
\end{align*}
%
For $P \in \Pcal(\reals^n)$ and $u \in U$, let $\mathcal{D}_u := \{\hat x^{(i)}, u, \hat x_+^{(i)}\}_{i=1}^N$
be a dataset of i.i.d. data, where $\hat x^{(i)}$ and $\hat x_+^{(i)}$ have been sampled from $P$ and $\Tau(\cdot \mid \hat x^{(i)}, u)$, respectively. The empirical estimate of the CME is $\hat\mu_u:\reals^n\rightarrow\mathcal{H}_k$ given by
\begin{align}
\label{eq:cme_emp}
    &\hat\mu_u(x)(\cdot) := \sum_{i = 1}^N \beta_i(x) k(\cdot, \hat x_+^{(i)}),
\end{align}
where $\beta(x) := (K_{\hat x \hat x} + N\lambda I)^{-1} k_{\hat x}(x)$ with $K_{\hat x, \hat x}$ being the kernel matrix of $\{\hat \x^{(i)}\}$, $k_{\hat x}(x) := [k(x, \hat x^{(1)}), \cdots, k(x, \hat x^{(N)})]^T$, and $\lambda \ge 0$.

Once empirical CME $\hat\mu_u$ is obtained, we use a concentration inequality that bounds the ``distance'' between each $\hat\mu_u$ and $\mu_u$. Such bounds exist in the literature for CMEs \cite{li2022optimal, grunewalder2012modelling,talwai2022sobolev}, ME of probability distributions \cite{gretton2012kernel}, GPs \cite{lederer2019uniform,reed2025learning}, and learning of probability distributions using the Wasserstein distance \cite{boskos2023high}.

\begin{assumption}
\label{ass:bound}
    Let the empirical CME $\hat\mu_u$ be obtained by sampling $\Dcal_u$, with $N := |\Dcal_u|$, from the $N$-fold joint distribution of $P$ and $\Tau$, for each $u \in U$. Then, the CME $\mu_u$ of the transition kernel in~\eqref{eq:sys} belongs to the RKHS $\HK$ for all $u\in U$. Furthermore, there exists a distribution $P \in \Pcal(\reals^n)$ and a function $(\epsilon, \delta)\mapsto N_{\epsilon, \delta} \in \naturals$ such that, for all $\varepsilon > 0, \delta \in (0,1)$, 
    if $N \ge N_{\varepsilon, \delta}$, then
    \begin{align}
        \forall  u \in U, \quad \|\mu_u - \hat\mu_u \|_{\HK} \le \varepsilon
    \end{align}
    with confidence at least $1- \delta$ over the choice of the samples.
\end{assumption}
The first statement in Assumption~\ref{ass:bound} is commonly referred to as the \emph{well-specified setting}, which ensures that we can learn the CME as a function in $\HK$. However, existing explicit results of the form Assumption~\ref{ass:bound} for CMEs are either not finite sample (but asymptotic) or involve quantities that are hard to bound. In Section~\ref{sec:explicit_CI}, we show how to obtain a similar bound under some assumptions on $\Tau$.

We now state a technical assumption needed to restrict our attention to a compact subset of $\reals^n$. The necessity for this assumption becomes clear in Section~\ref{sec:abstraction}.
\begin{assumption}[Compact set]
\label{ass:compact}
    For all $x\in X_\safe$ and $u\in U$, $\Tau(\cdot \mid x, u)$ is supported on some compact set $X \subset \reals^n$.
\end{assumption}
Note that the previous assumption 
Assumption~\ref{ass:compact} is mild, as it is trivially satisfied when the support of $\Tau(\cdot \mid x, u)$ is bounded. Otherwise, one can always pick $X_\safe$ small and $X$ big enough so that the probability that $\x_{t+1} \notin X$ is arbitrarily small.

\subsection{Uncertain MDPs}

%

We use UMDP as an abstraction of the stochastic process.

\begin{definition}[UMDP]
    \label{def:UMDP}
    A UMDP is a tuple 
    $\Ucal := (S,A,\Gamma)$ 
    in which $S$ and $A$ 
    are respectively finite set of states and actions, and $\Gamma := \{\Gamma_{s,a} \subseteq \mathcal{P}(S) : s\in S, a \in A\}$, where $\Gamma_{s,a}$ is the set of transition probability distributions, or \textit{ambiguity set}, from state $s$ under $a$.%
    to the states in $S$.
\end{definition}
%

A finite \emph{path} of UMDP $\Ucal$ is a sequence $\pathrmdp = s_0 \xrightarrow{a_0} s_1 \xrightarrow{a_1} \ldots 
\xrightarrow{a_{T-1}} s_T$ 
of states $s_t\in S$ and actions $a_t \in A$ such that there 
exists a distribution $\gamma\in\Gamma_{s_t,a_t}$ with $ \gamma(s_{t+1})> 0$ for all $t \in \{0, \dots, T-1\}$. We denote by $\Pathrmdp$ the set of all finite paths. Given path $\pathrmdp\in\Pathrmdp$, $\pathrmdp(t) = s_t$ is the state of $\pathrmdp$ at time $t$, and we denote its last state by $\last(\pathrmdp)$. 
A \emph{strategy} of a UMDP $\Ucal$ is a function $\strategy: \Pathrmdp \rightarrow A$ that maps each finite path to the next action. We denote by $\Strategy$ 
the set of all strategies of $\Ucal$. 
Given path $\pathrmdp \in \Pathrmdp$ and $\strategy \in \Strategy$, the process evolves from $s_{t} = \last(\path)$ under $a_t = \strategy(\pathrmdp)$ to the next state according to a probability distribution in $\Gamma_{s_t,a_t}$. An adversary is a function that chooses this distribution \cite{givan2000bounded}. Formally, an \emph{adversary} is a function $\adversary: S \times A \times \big(\naturals\cup \{0\}\big) \rightarrow \mathcal{P}(S)$ that maps each state $s_t$, action $a_t$, and time step $t\in\naturals\cup \{0\}$ to a transition probability distribution $\gamma\in\Gamma_{s_t,a_t}$, according to which $s_{t+1}$ is distributed. We denote the set of all adversaries by $\Adversary$. Given an initial condition $s_0\in S$, a strategy $\strategy\in\Strategy$ and an adversary $\xi\in\Xi$, the UMDP collapses to a Markov chain
with a unique probability distribution $Pr_{s_0}^{\strategy,\xi}$ over its paths.
\section{UMDP Abstraction}
\label{sec:abstraction}

In this section, we describe in detail how to obtain a correct UMDP abstraction of System~\eqref{eq:sys} using its empirical CME.

\subsection{Explicit Form of Assumption~\ref{ass:bound}}
\label{sec:explicit_CI}

First, we show how to obtain an explicit bound, analogous to the one in Assumption~\ref{ass:bound}, Assumption~\ref{ass:bound} when the transition kernel $\Tau$ arises from a dynamical system with a Lipschitz vector field, 
a common assumption in the literature.
Nevertheless, we emphasize our framework applies to any $\Tau$ under Assumption~\ref{ass:bound}.

Consider a stochastic system of the form $\x_{t+1} = f(\x_t, \w_t)$ with $f:\reals^n\times W\rightarrow \reals^n$ and $(\w_t)_{t\in \naturals}$ being i.i.d. random variables each distributed according to $P_w\in\Pcal(W)$. 
For simplicity and without loss of generality, we consider a stochastic system with no control, and the Gaussian kernel $k(x,x') = \sigma_f^2\exp(-\|x - x'\|^2/(2\sigma_l^2))$. We first state the following technical lemma:
\begin{lemma}
\label{lemma:explicit_bound}
    Let $f(x,w)$ be $L$-Lipschitz continuous in $x$
    and let $k$ be the Gaussian kernel. Then, for all $x, \tilde x \in X$ such that $\|x - \tilde x\| \le \eta$, with $\eta > 0$, it holds that $\MMD(\Tau(\cdot \mid \tilde x), \Tau(\cdot \mid x)) \le \frac{\sigma_f}{\sigma_l}L\eta$.
\end{lemma}
\begin{proof}
    By \cite{villani2021topics}, the Wasserstein distance between $\Tau(\cdot \mid \tilde x)$ and $\Tau(\cdot \mid x)$ is bounded by $\mathcal{W}(\Tau(\cdot \mid \tilde x), \Tau(\cdot \mid x)) \le L\eta$. Next, by definition of the MMD distance, $\MMD(\Tau(\cdot \mid \tilde x), \Tau(\cdot \mid x)) := \mid \int_{X}k(\cdot,y)\Tau(dy \mid x) - \int_{X}k(\cdot,\tilde y)\Tau(d\tilde y \mid \tilde x) \mid$. Let $\Pi \subseteq \mathcal{P}(X\times X)$ be the set of couplings between $\Tau(\cdot \mid x)$ and $\Tau(\cdot \mid \tilde x)$, and pick $\pi \in \Pi$. Then, by measurability and Bochner integrability of $k(\cdot,y)\in \Hk$, we can write $\int_{X}k(\cdot,y)\Tau(dy \mid x) = \int_{X\times X} k(\cdot,y)d\pi(y,\tilde y)$ and $\int_{X}k(\cdot,\tilde y)\Tau(d\tilde y \mid \tilde x) = \int_{X\times X}k(\cdot,\tilde y)d\pi(y,\tilde y)$. Therefore, $\MMD(\Tau(\cdot \mid \tilde x), \Tau(\cdot \mid x)) = \|\int_{X\times X}\big( k(\cdot,y) - k(\cdot,\tilde y) \big) d\pi(y,\tilde y) \|_{\Hk} \le \int_{X\times X}\|  k(\cdot,y) -k(\cdot,\tilde y)\|_{\Hk}  d\pi(y,\tilde y)$ for all couplings $\pi$, so it also holds if we take the $\inf$ over $\pi$.
    Next, by the reproducing property of $\Hk$ and, taking into account that $k$ is the Gaussian kernel, it holds that $\| k(\cdot,y) -k(\cdot,\tilde y)\|_{\Hk}^2 = k(y, y) + k(\tilde y, \tilde y) -2 k(y, \tilde y) \le \frac{\sigma_f^2}{\sigma_l^2}\|x - y\|^2$. Using this expression in the previous integral yields
    %
    $\MMD(\Tau(\cdot \mid \tilde x), \Tau(\cdot \mid x)) \le \inf_{\pi \in \Pi
    } \int_{X\times X} \frac{\sigma_f}{\sigma_l} \|y - \tilde y \| d\pi(y,\tilde y) =: \frac{\sigma_f}{\sigma_l} \mathcal{W}(\Tau(\cdot \mid \tilde x), \Tau(\cdot \mid x)) \le \frac{\sigma_f}{\sigma_l} L\eta.
    $
\end{proof}

\begin{theorem}
\label{thm:explicit_bound}
    Let $\x_{t+1} = f(\x_t, \w_t)$, 
    $f$ be $L$-Lipschitz in its first argument, $k$ be the Gaussian kernel. Let $\widetilde X$ be a discretization of $X$, i.e., a set of $M := |\widetilde X|$ points $\tilde x \in X$, with parameter $\eta$. Let the dataset in Problem~\ref{prob:problem} be of the form $\Dcal := \{(\tilde x_j, \x_{i,j}')\}_{i , j = 1}^{N,M}$, namely, where the successor of each $\tilde x \in \widetilde X$ is sampled $N$ times by starting at $\tilde x$. 
    Pick $\delta \in (0,1)$ and let 
    $$\varepsilon := (\sigma_f/\sigma_l) \Big( L \eta + \sqrt{1/N} + \sqrt{2\log(M/\delta)/N} \Big).$$ 
    Then, with probability great than or equal to $1-\delta$, 
    %
        $$\MMD(\Tau(\cdot \mid x), \Tau(\cdot \mid \tilde x)) \le \varepsilon \;\quad \forall x\in X, \forall \tilde{x} \in \widetilde X  \; \text{ s.t. }\; \|x - \tilde x\| \le \eta.$$
    %
\end{theorem}
\begin{proof}
    Let $\tilde x \in \widetilde X$ and $x\in X$ such that $\|x - \tilde x\| \le \eta$, and let $\widehat \Tau(\cdot \mid \tilde x)$ denote the empirical distribution corresponding to the samples $\{(\tilde x_i, \x'_{i,j})\}_{j = 1}^N$, with $\tilde x_i = \tilde x$. By \cite{gretton2012kernel, nemmour2022maximum}, $\MMD(\Tau(\cdot \mid \tilde x), \widehat \Tau(\cdot \mid \tilde x)) \le \sigma_f/\sigma_l\big(1/\sqrt{N} + \sqrt{2\log(M/\delta)/N}\big)$, with confidence not smaller than $1 - \delta/M$. By Lemma~\ref{lemma:explicit_bound}, $\MMD(\Tau(\cdot \mid x), \Tau(\cdot \mid \tilde x)) \le (\sigma_f/\sigma_l) L \eta$ for all $x \in X$ s.t. $\|x - \tilde x\| \le \eta$. The triangle inequality yields $\MMD(\Tau(\cdot \mid x)), \widehat \Tau(\cdot \mid \tilde x)) ) \le \MMD(\Tau(\cdot \mid x), \Tau(\cdot \mid \tilde x) ) + \MMD( \Tau(\cdot \mid \tilde x), \widehat \Tau(\cdot \mid \tilde x)) \le (\sigma_f/\sigma_l) L \eta + \epsilon$. 
    By picking $\varepsilon = (\sigma_f/\sigma_l) L \eta + \epsilon$, this result implies that for all $\tilde x \in \widetilde X$ and $x \in X$ with $\|x - \tilde x\| \le \eta$, $\MMD(\Tau(\cdot \mid x)), \widehat \Tau(\cdot \mid \tilde x)) ) \le \varepsilon$ holds with confidence $\ge 1-\delta/M$. Taking the union bound, the statement is extended to all $x\in X$ with an overall confidence $\ge1 - \delta$.
\end{proof}

We highlight that, even though this bound is slightly different from the one in Assumption~\ref{ass:bound}, it is sufficient to obtain a sound abstraction. Additionally, while this bound might be conservative and requires the strict assumption that the samples are prompted at deterministic points $\tilde x$, it provides intuition on what kind of dynamics satisfy Assumption~\ref{ass:bound} and it gives an explicit bound.

\subsection{Uncertain MDP Abstraction}


Here, we rely on the empirical CME and Assumption~\ref{ass:bound} 
to construct a sound finite UMDP abstraction. 
This construction is not limited to the explicit bounds derived in Section~\ref{sec:explicit_CI}; it generalizes to any system that satisfies Assumption~\ref{ass:bound}.

Our abstraction is based on a $\varphi_x$-conservative partition of set $X$.
\begin{definition}[$\varphi_x$-Conservative Partition]
    A finite partition $S = \{s_1, \ldots, s_{|S|-1}, s_\text{avoid}\}$ of $X$ is $\varphi_x$-conservative 
    if 
    (i) $\cup_{s\in S_\safe} s \subseteq  \Xsafe$, where $S_\safe := \{s_1, \dots, s_{|C|-1}\}$, 
    (ii) $s_\text{avoid} \supseteq \Xavoid$ and
    (iii) there exists a maximal subset $S_\text{reach} \subseteq S_\safe$ s.t. $\cup_{s \in S_\text{reach}} s \subseteq \Xreach$.
\end{definition}

\begin{definition}[UMDP Abstraction]
    \label{def:umdp_abstraction}
    A UMDP abstraction of System~\eqref{eq:sys}
    is $\Ucal := (S,A,\Gamma)$, where
    $S$ is a $\varphi_x$-conservative partition of set $X$,
    $A := U$, and $\Gamma := \{\Gamma_{s,a} \mid s\in S, a\in A\}$ such that
    \begin{multline}
        \label{eq:Gamma}
        \Gamma_{s,a} :=
        \big\{\gamma \in \Pcal(S) \mid \big \| \sum_{s'\in S} \gamma(s') k(\cdot, c_{s'}) - \hat\mu_a(c_s)\big\|_{\Hk} \le \varepsilon \big\},
        \end{multline}
    where $c_s \in X$ denotes the center point of the region $s$ and $\varepsilon := \varepsilon_1 + \varepsilon_2 + \varepsilon_3$ with $\varepsilon_1 > 0$,
    \begin{subequations}
        \begin{align}
        \varepsilon_2 &:= \max \big\{ \|\hat \mu_a(x)-\hat \mu_a(c_s)\|_{\mathcal{H}_k} \mid x \in s, s\in S \big\}, \label{eq: epsilon 2} \\
            \varepsilon_3 &:= \max\big\{\sqrt{ k(c_s, c_s) + k(y, y) - k(c_s, y)} \mid y \in s, s\in S  \big\}.
            \label{eq: epsilon 3}
        \end{align}
    \end{subequations}
\end{definition}

In the next subsection we describe an exact procedure to obtain $\varepsilon_2$ and also analyze the tightness of $\varepsilon_3$.

\subsubsection{Computation and Analysis of $\varepsilon_2$ and $\varepsilon_3$}

    First, we show how to compute $\varepsilon_2$ in 
    Equation~\eqref{eq: epsilon 2} exactly.

\begin{proposition}
Let $\Ucal = (S,A,\Gamma)$ be an abstraction of System~\eqref{eq:sys} per Definition~\ref{def:umdp_abstraction} and $k$ be the Gaussian kernel. Then,
    \begin{align}
        \label{eq:eps2_bound11}
            \varepsilon_2 = \max_{x\in s, s\in S}\sqrt{ (\beta(x) - \beta(c_s))^T K^{(3)} (\beta(x) - \beta(c_s)) },
    \end{align}
    where $K^{(3)}_{i,j} := k(\hat x_+^{(i)}, \hat x_+^{(j)})$, and the objective function
    in~\eqref{eq:eps2_bound11} is $\| \hat \mu_a \|_{\mathcal{H}_K} \sigma_f / \sigma_l$-Lipschitz.
    \end{proposition}
    \begin{proof}
By the definition of the empirical CME in~\eqref{eq:cme_emp}, $\| \hat\mu_a(x) - \hat\mu_a(c_s) \|_{\Hk} = \| \sum_{i = 1}^N (\beta_i(x) - \beta_i(c_s)) k(\cdot, \hat x_+^{(i)}) \|_{\Hk}$, which yields the objective in Problem~\eqref{eq:eps2_bound11}, that we denote by $g_s(x)$. Next, we obtain the Lipschitz constant of $g_s$. Note that $\|\hat \mu_a(x)-\hat \mu_a(c_s)\|_{\mathcal{H}_k} = \sup_{\|v\|_{\mathcal{H}_k} \le 1} \langle\hat \mu_a(x)-\hat \mu_a(c_s), v \rangle_{\mathcal{H}_k}$ and
\begin{align*}
    &\langle\hat \mu_a(x)-\hat \mu_a(c_s), v \rangle_{\mathcal{H}_k} = \langle\hat \mu_a(x), v \rangle_{\mathcal{H}_k} - \langle\hat \mu_a(c_s), v \rangle_{\mathcal{H}_k}\\
    &= \langle\hat \mu_a, K(\cdot, x) v\rangle_{\mathcal{H}_K} - \langle\hat \mu_a, K(\cdot, c_s) v\rangle_{\mathcal{H}_K} = \langle\hat \mu_a, K(\cdot, x)v - K(\cdot, c_s)v\rangle_{\mathcal{H}_K}\\
    &= \langle\hat \mu_a, (K(\cdot, x) - K(\cdot, c_s)) v\rangle_{\mathcal{H}_K} \le \| \hat \mu_a \|_{\mathcal{H}_K} \| (K(\cdot, x) - K(\cdot, c_s) v \|_{\mathcal{H}_K}.
\end{align*}

Next, we leverage the \emph{isometric isomorphism} between $\HK$ and the \emph{tensor product space} $\mathcal{H}_k \otimes \mathcal{H}_k$, where the norm of some $f\otimes g$, with $f, g\in\Hk$, is $\|f\otimes g\|_{\mathcal{H}_k \otimes \mathcal{H}_k} := \|f\|_{\Hk} \|g\|_{\Hk}$. By this isomorphism,
\begin{align*}
    &\| (K(\cdot, x) - K(\cdot, c_s)) v \|_{\HK} = \| v \otimes (k(\cdot,x) - k(\cdot,c_s))  \|_{\mathcal{H}_k \otimes \mathcal{H}_k}\\
    &= \|v\|_{\Hk} \| k(\cdot,x) - k(\cdot,c_s)\|_{\mathcal{H}_k}
\end{align*}

Putting these results together, we obtain
\begin{equation}
\label{eq:lipschitz_global_opt}
\begin{split}
    g_s(x) &\le \| \hat \mu_a \|_{\mathcal{H}_K} \sqrt{k(x,x) + k(c_s,c_s) - 2k(x,c_s)}\\
    &\le \| \hat \mu_a \|_{\mathcal{H}_K} \sigma_f / \sigma_l\|x - c_s\|,
    \end{split}
\end{equation}
and thus $\|g_s(x) - g_s(x')\| \le \| \hat \mu_a \|_{\mathcal{H}_K} \sigma_f / \sigma_l\|x - x'\|$ by the triangle inequality.
    \end{proof}
    
    Despite $g_s$ being nonlinear in $x$, knowing its Lipschitz constant allows us to employ a \emph{global optimization} algorithm, like the \emph{Piyavskii–Shubert} one \cite{piyavskii1972algorithm}, which obtains a tight upper bound on the the global maximum. To obtain $\epsilon_2$, this algorithm needs to be used $|S|$ times. Alternatively, we could have bounded $\varepsilon_2$ by plugging expression~\eqref{eq:lipschitz_global_opt} into \eqref{eq:eps2_bound11}, which yields a convex optimization problem. However, since this bound relies on Cauchy-Schwartz's inequality, it is more conservative than the optimal value in~\eqref{eq:eps2_bound11}.

    Next, we discuss the tightness of $\varepsilon_3$.

    \begin{remark}
    Note that the bound $\varepsilon_3$ is tight, as it can be attained. For example, let $\Tau(\cdot \mid x, a) = \delta_{x'}$ with $x'$ maximizing $\big(k(c_{s'}, c_{s'}) + k(x', x') - 2k(c_{s'}, x')\big)^{1/2}$ over all $x' \in X$. Let $s' \ni x'$. Then, $P_\gamma = \delta_{c_{s'}}$, and the MMD distance between the two distributions is exactly $\varepsilon_3$, as
    $$\|\Psi(\delta_{c_{s'}}) - \Psi(\delta_{x'})\|_{\Hk} = \sqrt{k(c_{s'}, c_{s'}) + k(x', x') -2k(c_{s'}, x')}.$$
\end{remark}

\subsection{Soundness of The UMDP Abstraction}

We now formalize soundness of abstraction $\Ucal$ in Definition~\ref{def:umdp_abstraction}, i.e., that $\Ucal$ contains all the behaviors of System~\eqref{eq:sys} (w.r.t. $\varphi_x$).
\begin{definition}[Soundness]
    \label{def:soundness}
    Let $\Ucal := (S,A,\Gamma)$ be a UMDP abstraction of System \eqref{eq:sys}. For each $s\in S\setminus \{s_\mathrm{avoid}\}$, $a\in A$, and $x \in s$, define $\gamma \in \Pcal(S)$ such that $\gamma(s') := \Tau(s' \mid x,a)$. UMDP $\Ucal$ is a \textit{sound abstraction} if 
    (i) $\gamma \in \Gamma_{s,a}$ for all $s\in S\setminus \{s_\mathrm{avoid}\}$ and $a\in A$,
    and 
    (ii) $\Gamma_{s_\mathrm{avoid},a} = \{\delta_{s_\mathrm{avoid}}\}$ for all $a \in A$.
\end{definition}

\begin{theorem}[Soundness]
    \label{thm:soundness}
    Consider Assumptions~\ref{ass:bound} and ~\ref{ass:compact}, and let $\varepsilon_1$ in Definition~\ref{def:umdp_abstraction} be equal to $\varepsilon$ in Assumption~\ref{ass:bound}. Then, the UMDP abstraction $\Ucal := (S,A,\Gamma)$ constructed per Definition.~\ref{def:umdp_abstraction}
    is sound.
\end{theorem}
\begin{proof}
    First we pick $s\in S$, $a\in A$, $x \in s$ and $s' \in S\setminus \{s_\mathrm{avoid}\}$, and define $P_\gamma := \sum_{s'\in S} \gamma(s')\delta_{c_{s'}}$. Then, note that $\gamma \in \Gamma_{s,a}$ holds if and only if $\| \Psi(P_\gamma) - \hat\mu_a(c_s)\|_{\Hk} \le \varepsilon$. We now show this is the case by upper bounding $\| \Psi(P_\gamma) - \hat\mu_a(c_s)\|_{\Hk}$ by $\| \Psi(P_\gamma) - \Psi(\Tau(\cdot \mid x, a))\|_{\Hk} + \| \Psi(\Tau(\cdot \mid x, a)) - \hat\mu_a(x)\|_{\Hk} + \| \hat\mu_a(x) - \hat\mu_a(c_s)\|_{\Hk}$ via the triangle inequality. By definition of the ME operator and Assumption~\ref{ass:bound}, the second term is bounded by $\| \Psi(\Tau(\cdot \mid x, a)) - \hat\mu_a(x)\|_{\Hk} = \| \mu_a(x') - \hat\mu_a(x)\|_{\Hk} \le \varepsilon_1$. Similarly, since $x \in s$, we obtain that the third term is upper bounded by $\varepsilon_2(s)$. Finally, we show that $\| \Psi(P_\gamma) - \Psi(\Tau(\cdot \mid x, a))\|_{\Hk} \le \varepsilon_3$. Note that, by definition of the ME map and linearity of the Bochner integral,
    \begin{align*}
        \| \Psi&(P_\gamma) - \Psi(\Tau(\cdot \mid x, a))\|_{\Hk}\\
        = &\| \int_{X} k(\cdot, x')P_\gamma(dx') - \int_{X} k(\cdot, x')\Tau(dx' \mid x, a)\|_{\Hk}\\
        = &\| \sum_{s'\in S}k(\cdot, c_{s'})\gamma(s') - \sum_{s'\in S}\int_{s'} k(\cdot, x')\Tau(dx' \mid x, a)\|_{\Hk}\\
        = &\| \sum_{s'\in S}k(\cdot, c_{s'})\int_{s'} \Tau(dx' \mid x, a) - \sum_{s'\in S}\int_{s'} k(\cdot, x')\Tau(dx' \mid x, a)\|_{\Hk}\\
        \overset{(a)}{\le} &\sum_{s'\in S}\int_{s'} \| k(\cdot, c_{s'})  - k(\cdot, x') \|_{\Hk}\Tau(dx' \mid x, a)\\
        \le &\sum_{s'\in S}  \Tau(s' \mid x, a)  \max\big\{\| k(\cdot, c_{s'})  - k(\cdot, x') \|_{\Hk} \mid x' \in s'  \big\}\\
        \le &\max\big\{\sqrt{ k(c_{s'}, c_{s'}) + k(x', x') - 2k(c_{s'}, x')} \mid x' \in s', s'\in S  \big\} = \varepsilon_3,
    \end{align*}
    where $(a)$ follows from the triangle and Jensen's inequalities.
\end{proof}

\section{Control Synthesis}
\label{sec:synthesis}

Here, we focus on generating a strategy over the sound abstract $\Ucal$ that is optimally robust w.r.t. all uncertainties in the abstraction. We then refine it to a policy for System~\eqref{eq:sys} with guarantees for  $\varphi_x$.

\subsection{Robust Dynamic Programming (RDP)}

Consider UMDP $\Ucal$, a reach-avoid specification $\varphi = (\Sreach, \{s_\text{avoid}\})$ with $\Sreach,\{s_\text{avoid}\}\subseteq S$, a strategy $\strategy\in \Strategy$, and an adversary $\xi\in\Xi$. With a slight abuse of notation, we denote by $\text{Pr}_{s}^{\strategy,\xi}[\varphi]$ the probability that the paths of $\Ucal$ starting at $s \in S$ satisfy $\varphi$ under $\strategy$ and $\xi$, which is defined analogously to \eqref{eq: satisfaction_prob}.
\begin{proposition}[RDP {\cite[Theorem 6.2]{gracia2025efficient}}]
\label{prop:robust_value_iteration}
Given a UMDP $\Ucal = (S,A,\Gamma)$, a reach-avoid specification $\varphi = (\Sreach, \{s_\text{avoid}\})$, and $s \in S$, 
define the optimal robust reach-avoid probability of $\Ucal$ as 
$\underline p(s) :=\sup_{\strategy \in\Strategy}\inf_{\xi\in\Xi} \text{Pr}_{s}^{\strategy, \xi}[\varphi]$. Consider also the recursion
\begin{align}
    \label{eq:rdp_lower_bound}
    \underline p^{t+1}(s) =
    \begin{cases}
    \max\limits_{a\in A}\min\limits_{\gamma\in\Gamma_{s,a}} \sum\limits_{s'\in 
             S}\gamma(s')\underline p^{t}(s') & \text{if } s \in S\setminus \Sreach\\
    1   & \text{otherwise}
    \end{cases}
\end{align}
%
where $t\in\naturals_0$, with initial condition $\underline p^{0}(s) = 1$ if $ s\in\Sreach$ otherwise $0$. Then, $\underline p^t$ converges to $\underline p$.
\end{proposition}
The most computationally intensive part when iterating on~\eqref{eq:rdp_lower_bound} is solving the inner minimization problems over the sets $\Gamma_{s,a}$. The following theorem shows that these minimizations are quadratically-constrained linear programs (QCLPs), which can be solved using off-the-shelf convex solvers.

\begin{theorem}
    For every $s \in S\setminus\{s_\mathrm{unsafe}\}$ and $a \in A$, the inner minimization problem over the ambiguity set $\Gamma_{s,a}$ in \eqref{eq:rdp_lower_bound} is equivalent to the finite-dimensional convex program 
\begin{subequations}
    \label{eq:inner_problem_finite}
    \begin{align}
        \min_{\gamma \ge 0} \quad & \sum_{i = 1}^{|S|} \gamma_i \underline p^k(s_i)\\
        \textrm{s.t.} \quad & \gamma^T K_1\gamma  - 2\gamma^T K_2 \beta + \beta^T K_3 \beta \le \varepsilon^2, \sum_{i = 1}^{|S|} \gamma_i = 1
        \end{align}
\end{subequations}
with $K_{i,j}^{(1)} = k(c_{s_i}, c_{s_j}) \:\forall i,j \in \{1, \dots, |S|\}$,  $K_{i,j}^{(2)} = k(c_{s_i}, \hat x_j) \:\forall i \in \{1, \dots, |S|\}, j \in \{1, \dots, N\}$ and $K_{i,j}^{(3)} = k(\hat x_i, \hat x_j) \:\forall i,j \in \{1, \dots, N\} $.
\end{theorem}
\begin{proof}
    Note that the simplex constraints, i.e., non-negativity of $\gamma$ and $\sum_{i = 1}^{|S|} \gamma_i = 1$ are equivalent to $\gamma \in \Pcal(S)$ in \eqref{eq:Gamma}. Furthermore, pick an arbitrary $\gamma \in \Gamma_{s,a}$. Then, $\|\sum_{s'\in S} \gamma(s') k(\cdot, c_{s'}) - \hat\mu(c_s)\|_{\Hk} \le \varepsilon$, which is 
    equivalent to the quadratic constraint in \eqref{eq:inner_problem_finite}. Finally, since $K^{(1)}$ comes from a positive semi-definite kernel is itself positive semi-definite we conclude that the constraint is convex.
\end{proof}

\textbf{Obtaining optimal strategy and policy from RDP. }
Based on the RDP in \eqref{eq:rdp_lower_bound}, work \cite[Theorem 6.4]{gracia2025efficient} introduces a polynomial algorithm to obtain a stationary strategy $\strategy^*$ that is both optimal and robust, i.e., $\strategy^*(s) 
\in \arg\max_{\strategy \in\Strategy}\inf_{\xi\in\Xi} \text{Pr}_{s}^{\strategy, \xi}[\varphi]$ for all $s\in S$. Then from $\strategy^*$, we obtain the optimistic probabilities $\overline p(s) := \sup_{\xi\in\Xi} \text{Pr}_{s}^{\strategy^*,\xi}[\varphi]$ by iterating on the recursion in \cite[Equation 6.5]{gracia2025efficient}, which is similar to the one in \eqref{eq:rdp_lower_bound} but the actions are determined by $\strategy^*$ and the $\min$ over $\Gamma_{s,a}$ is replaced by a $\max$. Finally, we map $\strategy^*$ to policy $\policy^*$ of System~\eqref{eq:sys} as $\policy^*(x) := \strategy^*(s)$ with $ s \ni x$.

\subsection{Correctness of RDP}

The following theorem gives a bound in the probability that System~\eqref{eq:sys} satisfies $\varphi_x$ under $\pi^*$, 
solving Problem~\ref{prob:problem}. The proof is analogous to that of \cite[Theorem 2]{skovbekk2021formal} and a consequence of Theorem~\ref{thm:soundness}.
\begin{theorem}[Correctness]
\label{thm:strategy_synthesis}
    Let $\underline p$ be the solution of the RDP recursion \eqref{eq:rdp_lower_bound}, $\strategy^*$ and $\overline p$  be obtained per \cite[Theorem 6.4 \& Equation 6.5]{gracia2025efficient}, and $\policy^*$ be the refinement of $\strategy^*$ to System~\eqref{eq:sys}. Then, for all $x\in X$, it holds that $\text{Pr}_{x}^{\policy^*}[\varphi_x] \in [\underline p(s), \overline p(s)]$ with $s\in S$ such that $x \in s$.
\end{theorem}

\section{Experimental Evaluation}
\label{sec:experiments}

We showcase the effectiveness of our approach to produce policies for unknown systems with formal satisfaction guarantees through the temperature regulation benchmark in~\cite{romao2023distributionally}. We consider two specifications: (i) safety, where the system's temperature is to be kept within safe bounds for $15$ time steps, and (ii) unbounded reach-avoidance, where the temperature is also to be regulated to the interval $[20.25, 20.75]$ within no specific time horizon and while remaining safe. We compare our safety results with the approach of \cite{romao2023distributionally}. 
%
We validate our theoretical results through $500$ Monte Carlo simulations starting at each state.

The system is given in \cite{romao2023distributionally}, with $X_\safe = [19, 22]$. The noise is Gaussian with standard deviation of $0.15$, truncated to satisfy Assumption~\ref{ass:compact} with $X = [17.5, 23.5]$. The kernel is Gaussian with $\sigma_l = 1$, $\sigma_f = 10$, $N = 7000$ samples per control, and $\lambda = 0.001$. We obtain $S$ by uniformly partitioning $X$. 

The results of the safety problem are shown in Figure~\ref{fig:safety} and Table~\ref{tab:table}, where the latter shows results for different partitions and a comparison with \cite{romao2023distributionally}. Our method safely controls the system while providing formal bounds in the probability of safety, which are empirically validated. These bounds become tighter as $|S|$ increases and $\varepsilon_1$ decreases. However, this is not the case for the approach of \cite{romao2023distributionally}, where safety worsens with $|S|$. The reason is that \cite{romao2023distributionally} approximates the initial value function, which is discontinuous, by a continuous RKHS function. This causes its RKHS norm to increase with the number of approximation points, which decreases the safety probability, and the ability of the approach to produce a good policy. In addition to completely bypassing this issue, our approach yields formal performance guarantees. 

Finally, the results for the reach-avoid problem are given and empirically validated in Figure~\ref{fig:figure}, showing that our method yields high and tight bounds in the reach-avoid probability.



\begin{figure}[ht]
    \centering
    \begin{subfigure}{0.22\textwidth}
        \centering
        \includegraphics[width=\linewidth]{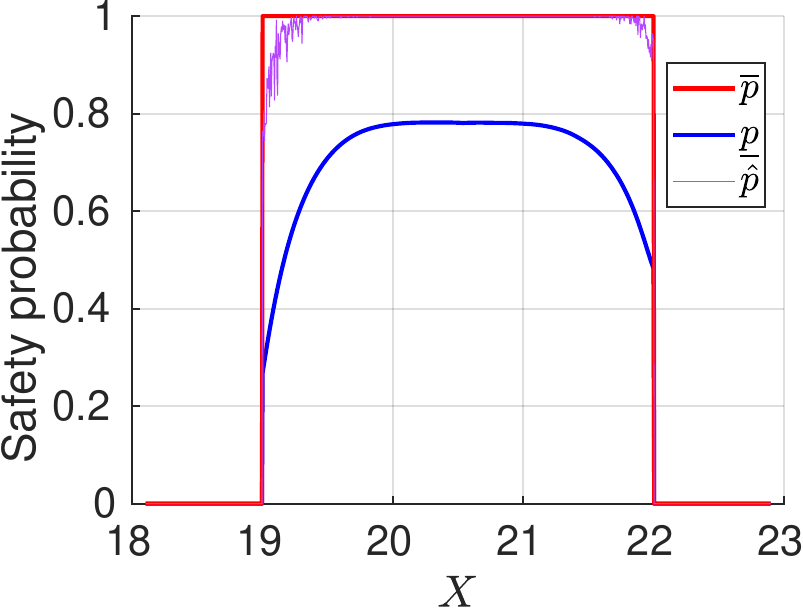}
        \caption{Safety}
        \label{fig:safety}
    \end{subfigure}
    \hfill
    \begin{subfigure}{0.22\textwidth}
        \centering
        \includegraphics[width=\linewidth]{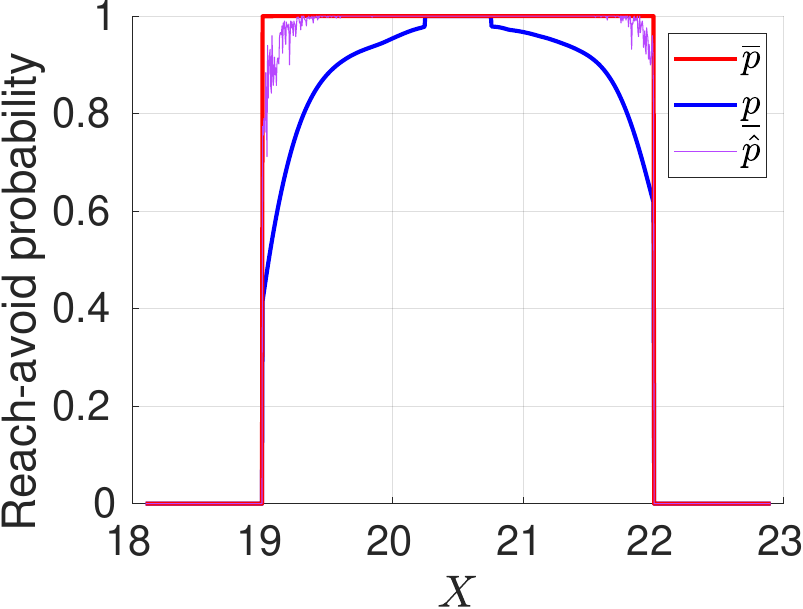}
        \caption{Reach-Avoid}
        \label{fig:reach}
    \end{subfigure}
    \caption{Theoretical bounds $\underline{p},\overline{p}$ and empirical probability $\hat p$ as a function of the temperature $x\in X$}
    \label{fig:figure}
\end{figure}

\begin{table}[]
    \centering
    \caption{Safety case study. $\underline p_{avg}$ denotes the average over $s_0\in S_\safe$ of $\underline p(s)$, 
    and $\underline p_{avg}$ is the average empirical satisfaction probability. Reported times are in minutes.}
    \label{tab:table}
    \scalebox{0.94}{
    \begin{tabular}{c|c c c c c c c}
    \toprule
            \multirow{2}{*}{Method} & \multirow{2}{*}{$|S|$} & \multirow{2}{*}{$\varepsilon_1$} & \multirow{2}{*}{$e_{avg}$} & \multirow{2}{*}{$\underline{p}_{avg}$} & \multirow{2}{*}{$\hat{p}_{avg}$} & Abstr. & Synth.
 \\
         &  &  &  & &  & Time & Time \\
        \hline
        Ours & $35$ & $0.09$ & $0.936$ & $0.064$ & $0.991$ & $0.539$ & $0.195$ \\
        \cite{romao2023distributionally} & &  & \text{N/A} & $0.250$ & $0.982$ & $0.134$ & $0.120$ \\
        \hline
        Ours & $1e2$ & $0.023$ & $0.700$ & $0.300$ & $0.992$ & $1.406$ & $0.494$ \\
        \cite{romao2023distributionally} & &  & \text{N/A} & $0.117$ & $0.988$ & $0.138$ & $0.355$ \\
        \hline
        Ours & $5e2$ & $0.034$ & $0.388$ & $0.612$ & $0.989$ & $6.899$ & $30.568$ \\
        \cite{romao2023distributionally} & &  & \text{N/A} & $0.01$ & $0.990$ & $0.1372$ & $1.840$ \\
        \hline
        Ours & $1e3$ & $0.035$ & $0.288$ & $0.712$ & $0.990$ &  $21.883$ & $220.126$ \\
        \cite{romao2023distributionally} & &  & \text{N/A} & $0$ & $0.610$ &$0.134$ & $3.914$ \\
    \bottomrule
    \end{tabular}
    }
\end{table}
\section{Conclusion}
\label{sec:conclusion}

We presented a framework for formal data-driven control of stochastic systems under complex specifications, via CMEs and UMDP abstractions. Our approach is the first one to obtain a sound finite abstraction of a CME. We empirically validated our method through a temperature regulation benchmark, obtaining better results than state-of-the-art CME-reliant approaches. A promising future research direction involves leveraging deep kernels to increase tightness of the abstraction to yield better results with finer discretizations.

\bibliographystyle{acm}
\bibliography{refs}

\end{document}